\documentclass{ifacconf}

\usepackage{natbib}            
\usepackage{latexsym,amssymb, color,epsfig,mathrsfs,amsmath,amssymb}
\usepackage{indentfirst,latexsym,bm}
\usepackage{subfigure,graphicx,graphics,color}
\usepackage{mathtools}
\usepackage[]{mdframed}
\usepackage{enumerate}
\usepackage{comment}

\newtheorem{proposition}{Proposition}

\newtheorem{lemma}{Lemma}
\newtheorem{definition}{Definition}

\newenvironment{proof}{\medskip\noindent{\it Proof. }}{\hfill$\square$}
\newenvironment{example}{\medskip\noindent{\it Example. }}{\medskip}
\usepackage{tabularx,ragged2e,booktabs,caption}

\def\RR{{\mathbb R}}    

\def\bA{{\bf{A}}}
\def\bB{{\bf{B}}}
\def\bC{{\bf{C}}}

\def\bE{{\bf{E}}}

\def\bG{{\bf{G}}}

\def\bL{{\bf{L}}}

\def\bN{{\bf{N}}}

\def\bP{{\bf{P}}}

\def\bx{\textbf{x}}
\def\by{\textbf{y}}
\def\bz{\textbf{z}}

\begin{document}

\begin{frontmatter}

\title{Chaotic Masking Protocol for Secure Communication and Attack Detection in Remote Estimation of Cyber-Physical Systems}

 \author[First]{Tao Chen}
 \author[Second]{Andreu Cecilia}
 \author[Third]{Daniele Astolfi}
 \author[First]{Lei Wang}
 \author[First]{Zhitao Liu}
 \author[First]{Hongye Su}

 \address[First]{\small College of Control Science and Engineering, Zhejiang University, P.R. China (e-mail: tao\_chen;lei.wangzju;ztliu; hysu69@zju.edu.cn)}
 \address[Second]{\small Universitat Polit\'ecnica de Catalunya, Avinguda Diagonal, 647, 08028 Barcelona, Spain. (e-mail: andreu.cecilia@upc.edu)}
 \address[Third]{\small Univ. Lyon, Universit\'e Claude Bernard
Lyon 1, CNRS, LAGEPP UMR 5007, F-69100 Villeurbanne, France (e-mail: daniele.astolfi@univ-lyon1.fr)}

\thanks{This research is carried out within the framework of the EU-funded Recovery, Transformation and Resilience Plan (Next Generation) thorugh the project ACROBA and the Spanish Ministry of Universities funded by the European Union - NextGenerationEU (2022UPC-MSC-93823). This work is part of the Project MAFALDA (PID2021-126001OB-C31) funded by MCIN/ AEI /10.13039/501100011033 and by ``ERDF A way of making Europe''. This work is part of the project MASHED (TED2021-129927B-I00),  funded by MCIN/ AEI/10.13039/501100011033 and by the European Union Next GenerationEU/PRTR.
This work has been also partially founded by the ANR project Alligator (ANR-22-CE48-0009-01).
}


\begin{abstract}                          
 In remote estimation of cyber-physical systems (CPSs), sensor measurements transmitted through network may be attacked by adversaries, leading to leakage risk of privacy (e.g., the system state), and/or failure of the remote estimator. To deal with this problem, a chaotic masking protocol is proposed in this paper to secure the sensor measurements transmission. In detail, at the plant side, a chaotic dynamic system is deployed to encode the sensor measurement, and at the estimator side, an estimator estimates both states of the physical plant and the chaotic system. With this protocol, no additional secure communication links is needed for synchronization, and the masking effect can be perfectly removed when the estimator is in steady state. Furthermore, this masking protocol can deal with multiple types of attacks, i.e., eavesdropping attack, replay attack, and stealthy false data injection attack.
\end{abstract}
\end{frontmatter}

\section{Introduction}
\label{sec:intro}

Integrating technologies including control, communication, and computation, cyber-physical systems (CPSs) exhibit a substantial promise across diverse industry sectors, encompassing smart cities, smart grids, smart manufacturing, and intelligent transportation \citep{ZHANG20211,humayed2017cyber}. Nevertheless, the complex intricate nature of these systems, notably the network-based data transmission, exposes them to vulnerabilities and potential security breaches, leading to great damages \citep{Stuxnet2011,2008Lessons}.

A remote estimator uses data collected remotely to estimate the state of a physical system, enabling remote monitor and control.
There exist multiple attack strategies for the remote estimation problem. Some remarkable examples are eavesdropping attack \citep{shang2022linear}, replay attack \citep{liu2021active}, and false data injection (FDI) attack \citep{7470566}.

An eavesdropping attack aims to gain unauthorized access to the state information of a physical plant by intercepting sensor measurements transmitted through the network. So the privacy of the transmitted information should be protected to deal with this kind of attack.
Replay attack involves recording sensor data for a certain amount of time, and replaying in a loop the recorded data when needed. Since the recorded data is taken from the true plant, it is not easy to be detected by a monitor or detector. An  FDI attack aims to destroy the estimation by injecting specific attack sequence on the communicated sensor reading. In general, the attack sequence is designed such that they can bypass the detector, i.e., the attack is stealthy. Hence, for replay attack and FDI attacks, the method to detect them plays a fundamental role.

One commonly used detector is the innovation-based detector, which relies on the discrepancy between predicted and actual received sensor outputs to detect attacks \citep{6897944}. In \citep{pasqualetti2013attack}, the concept of detectable attack is built, and then necessary and sufficient conditions are presented for the innovation-based detector to detect detectable attacks. However, the authors in \citep{liu2021active} show that a innovation-based detector may fail to detect replay attacks in the sense that its detection rate converges to the false alarm rate when replay attack happens. In addition, if an attacker is aware of the detector and designs attacks carefully, the detector may also fail to detect an FDI attack. For example, \citep{7470566} proposed a linear deception attack strategy, which ensures that the innovation term with attack preserves the same Gaussian distribution as the healthy one and can successfully bypass the innovation-based detector.  \citep{hu2018state} presented a FDI attack that ensures a bounded innovation term but an unbounded estimation error.

To protect the transmitted data and assist stealthy attack detection of innovation-based detector, watermark-based and encode-decode-based methods are usually proposed.
 In general, watermarks are random noises added to the control input to increase the detection probability of a FDI attack \citep{liu2021active,naha2023quickest,satchidanandan2016dynamic}. In \citep{liu2021active}, both independent and identically distributed (i.i.d) Gaussian and non-i.i.d. Gaussian noises are used as watermark to detect replay attacks. In \citep{satchidanandan2016dynamic},  a watermark strategy is used for more general attack types, where an attack can be any malicious behaviors of sensors.
Nonetheless, since the watermark signal would increase control costs, a lot of efforts have been paid to decrease the effect of watermark \citep{naha2023quickest,fang2020optimal,chen2023replay}.

Alternatively, encode-decode methods use random numbers to encode the  information to be transmitted and decode it at the receiver's side, e.g., \citep{shang2022linear,guo2022detection,2020Resilient,miao2016coding}. In \citep{guo2022detection}, to detect stealthy FDI
attacks, a stochastic coding scheme is proposed. The encode is based on adding i.i.d. Gaussian noise to sensor measurements, and decode is by subtracting the same noise. In \citep{miao2016coding}, to detect stealthy FDI attacks, the transmitted information is multiplied by a coding matrix, which depends on the system information and a random generator. In \citep{shang2022linear}, linear encryption strategies are used
 to against eavesdropping attack by degrading its estimation performance.

  To ensure accurate decoding of sensor measurements, it is essential for the random encoding and decoding processes to be synchronized. To mitigate the risk of directly transmitting decoding information, the utilization of pseudo-random number generators is recommended. This approach necessitates the distribution of secret random seeds and synchronization between the generators employed on both the sensor and controller sides. In \citep{2020Resilient}, chaotic oscillators are used to generate coding and decoding signals at plant and estimator side, respectively.
However, there still needs a secure channel to synchronize the oscillators at the both sides.

In \citep{cecilia2023masking}, a quasi-periodic signal is used to encode (mask) the sensor measurement, and a washout filter \citep{wang2020pre} is used to remove the masking signal. With this masking protocol, the privacy is protected and a FDI attack in the communication channel can be detected. It is worth mentioning that this protocol does not need additional communication links for synchronization. Inspired by this result, in this work we use chaotic signal to encode the sensor measurement. To decode the sensor measurement, an estimator combining with the  dynamics of the chaotic system is designed.

The main contribution of this work is summarized as follows:
\begin{itemize}
    \item This paper proposes a novel chaotic masking protocol that secures the sensor-estimator communication channel. The protocol is proven to be effective against multiple attack types, including eavesdropping, replay attacks, and false data injection attacks.

    \item Similar to the coding-decoding mechanisms, the proposed protocol can  recover the sensor outputs completely, without compromising estimation performance. Notably, unlike most coding-decoding mechanisms, the proposed protocol does not require an additional secure communication channel to synchronize the masking and de-masking pair.
\end{itemize}

 \noindent{\bf Notation}. We denote by $\mathbb{R}$ the set of real numbers, $\mathbb{C}$ the set of complex numbers, $\mathbb{R}^n$  the set of real space of $n$ dimension for any positive integer $n$.
 For $a, b \in \mathbb{R}$, $a\underset{\mathcal{H}_0}{\overset{\mathcal{H}_1}{\gtrless}}b$ returns $\mathcal{H}_1$ if $a>b$, and  $\mathcal{H}_0$, if $a<b$. For a variable $x\in\mathbb{R}^n$, $x_{i}$ denotes the $i^{\rm th}$ element of $x$.
 $x^{\tau}$ denotes the delayed signal, i.e., $x_k^{\tau}:=x({t-\tau})$. Given $x\in\mathbb{R} $, $\sigma\in \mathbb{R}_{\geq 0} $, define $\text{sat}_\sigma(x):= \text{max}\{-\sigma,\text{min}\{\sigma,x\}\}$. Given $x\in\mathbb{R}^n $, $\sigma:=(\sigma_1,...,\sigma_n)\in \mathbb{R}^n_{\geq 0} $, define $\text{sat}_\sigma(x):= (\text{sat}_\sigma(x_1),...,\text{sat}_\sigma(x_n))$.
 For $A\in \mathbb{R}^{n\times n} $, $A^{\top}$ denotes the transpose of matrix $A$, $A$ is called a Hurwitz matrix if all of its eigenvalues have negative real parts, and $\rho_{\text{max}} (A) (\rho_{\text{min}}(A) )$ denotes the maximum (minimun) eigenvalue of a symmetric matrix $A$.
   $I_{n\times n}$ denotes identity matrix of dimension $n \times n$. For column vectors $x\in \mathbb{R}^m$ and $y\in \mathbb{R}^n$, $\text{col}(x,y): = [x^\top,y^\top]^\top$. For matrices $A\in\mathbb{R}^{m\times m}$ and $B\in\mathbb{R}^{n\times n}$, $\text{diag}(A,B)$ denotes a diagonal matrix with diagonal blocks $A$ and $B$.

\section{Problem Formulation}

 \subsection{Remote Estimator}

 Consider  linear systems of the form
\begin{equation}\label{eq.system_linear}
\Sigma_p:
\; \left\{
\begin{split}
  \dot{x}&=A{x}+Bu\\
  y  &= C{x},
\end{split}
\right.
 \end{equation}
 where $x \in \mathbb{R}^{n_x}$ is the system state, $u \in \mathbb{R}^{n_u}$ is the known input signal, and $y\in\mathbb{R}^{n_y}$ is the sensor measurement.
  For monitoring or control purposes, the sensor measurement $y$ is assumed to be transmitted through wireless networks to a remote estimator  of the form
\begin{equation}\label{eq.system_linear_obsv_ori}
\Sigma_{e}:
\; \left\{\begin{split}
  \dot{\hat x}&=A{\hat x}+Bu+ L( y-\hat{ y} )\\
  \hat y  &= C{\hat x}.
 \end{split}
 \right.
\end{equation}
It is well known that if $L$ is chosen such that $A-LC$ is Hurwitz, the estimation error converges to zero, exponentially.

 In the literature, a detector is commonly deployed after the estimator to detect abnormal data from the received sensor measurements. In this paper, the following innovation-based abnormal detector \cite{mehra1971innovations,pasqualetti2013attack} is deployed as
\begin{equation}\label{eq.detector}
  g(z ) := z^{\top}z\underset{\mathcal{H}_0}{\overset{\mathcal{H}_1}{\gtrless}}\nu
\end{equation}
 where $z =y-\hat{y} $, $\mathcal{H}_0$ means that no alarm is triggered, $\mathcal{H}_1$ means that the detector alarms, and $\nu$ is a threshold implemented to avoid false alarms induced by sensor and/or channel noise.

In this paper, we consider that the transmission network may be attacked by adversaries, leading to  leakage risk of private data (e.g., the system state), and/or  failure of the remote estimator by maliciously modifying the transmitted information. Furthermore, it may cause security issues if the estimated states are further utilized for control or other purposes.
Specifically, we are mainly interested in following three types of attacks.
\begin{itemize}
    \item[i)] {\bf Eavesdropping attack.} The adversary has access to  the transmitted information over the wireless network, and aims to eavesdrop the information of system states.
       \item[ii)] {\bf Replay attack.} The adversary records the normal sensor data and replay them in sensor-controller channel when injecting false data into the controller-actuator channel, destroying the plant while avoid being detected by detector \eqref{eq.detector}.
    \item[iii)] {\bf Stealthy data injection attack.} The adversary modifies the   the transmitted information over the wireless network by adding stealthy signals $\{a\}$ such that the remote estimator fails while avoid being detected by detector \eqref{eq.detector}.
\end{itemize}

 The main purpose of this paper is securing the communication channel by protecting the transmission signal $y$. Specially, the security refers to
\begin{itemize}
    \item[i)] Prevent eavesdropper from estimating $x$ by using $y$ and the information of the linear system $\Sigma_p$.
     \item[ii)] Invalidate the replay attack strategy.
     \item[iii)]  Destroy the stealthiness of the data injection attack.
\end{itemize}

\section{System Protection with Chaotic Masking}
\subsection{Masking Approach}
 To secure the communication channel, a chaotic signal is used to mask the sensor measurement $y$, i.e.,
\begin{equation*}
    \by:= y+d,
\end{equation*}
where $d$ is the chaotic signal to be designed. In this way, the masked sensor measurement $\by$ is transmitted in the channel, instead of transmitting $y$ directly.

 The chaotic signal is generated by a chaotic system of the following form
\begin{equation}\label{eq.chaotic_system}
\Sigma_c:
\; \left\{\begin{split}
  \dot{\xi}&=\Phi{\xi}+ \varphi(\xi)\\
  d &= \Lambda {\xi}
\end{split}\right.
\end{equation}
where $\xi \in \Xi\subset \bm\bar\Xi  \subset  \mathbb{R}^{n_\xi}$ is the state of chaotic system, with $\Xi$ being a compact forward invariant set for the chaotic dynamics. $\bm\bar \Xi$ is a hypercube containing $\Xi$, with the center of $\bm\bar \Xi$ being $0_{n_\xi}$ and the length of its $i^\text{th}$ edge being $2\sigma_i$, where \begin{equation}\label{eq.sigma}
\sigma_i:=  \max_\xi| \xi_i|,~\xi\in\Xi,
\end{equation}
which is illustrated in Fig. \ref{fig.Xi} in the case $n_\xi=2$.
It is assumed that the mapping $\varphi(\xi)$ is locally Lipschitz, that is
\begin{equation}\label{eq.lipschitz_g}
\begin{split}
  \|\varphi(\xi_1)-\varphi(\xi_2)\|\leq  \ell_\varphi  \|\xi_1-\xi_2\|,~~\forall \xi_1, \xi_2 \in \bm\bar\Xi
 \end{split}
\end{equation}
for some positive constant $ \ell_\varphi $, which is related to $\Xi$.
It is further assumed that the model of the chaotic system, i.e., $\Phi, \varphi(\cdot)$ and $\Lambda$ are known for remote estimator, but unknown to the potential adversaries.

Finally, by letting
 $\bx: =\text{col}( \xi,  x) $,
   the plant  $\Sigma_p$ and the chaotic masking $\Sigma_c$ are compactly rewritten as
\begin{equation}\label{eq.extended_eta}
\Sigma_{\bP}:
\; \left\{\begin{split}
 \dot{\bx}&=\bA\bx+\bB u+\bG(\xi)\\
\by & = \bC\bx.
     \end{split}
     \right.
\end{equation}
where we defined
\begin{equation*}
\begin{split}
\bA:=\left [ \begin{array}{cc}
 \Phi  &  0\\
0  & A \\
\end{array} \right],
\bB := \left [ \begin{array}{c}
0 \\  B\\
\end{array} \right],
\bG(\xi): = \left [ \begin{array}{c}
\varphi( \xi) \\ 0\\
\end{array} \right],
\bC:= \left [ \begin{array}{cc}
\Lambda  &  C\\
\end{array} \right].
 \end{split}
\end{equation*}
Note that since $\varphi(\xi)$ is locally Lipschitz, $\bG(\xi)$ is also locally Lipschitz with the same constant $ \ell_\varphi$.

\begin{figure}[t]
\centering
\includegraphics[width=0.38\columnwidth]{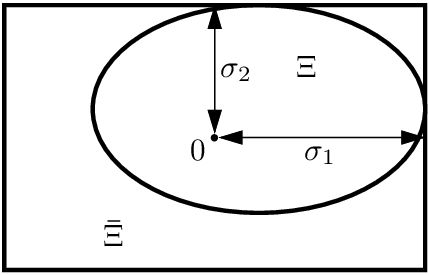}
 \caption{Representation of the set $\Xi$ and $\bm\bar\Xi$ in $\RR^2$.}
\label{fig.Xi}
\end{figure}

 The following example shows a chaotic system with form \eqref{eq.chaotic_system} and satisfying \eqref{eq.lipschitz_g}.

\begin{example}
(\textit{R$\ddot{o}$ssler prototype-4 system}) Consider the following system
 \begin{equation}\label{eq.rossler}
 \begin{split}
 \dot{{\xi}}_1 & = -{\xi}_2-{\xi}_3\\
 \dot{\xi}_2 &= \xi_1\\
 \dot{\xi}_3 &= a({\xi}_2-{\xi}_2^2)-b{\xi}_3
\end{split}
 \end{equation}
  which is chaotic when $a = b = 0.5$, as shown in Fig. \ref{fig.rossler}. This system can be rewritten in the form given by  \eqref{eq.chaotic_system}, where $\xi =\text{col}(\xi_1,\xi_2,\xi_3)$,
 \begin{equation*}
\begin{split}
 \Phi:=\left [ \begin{array}{ccc}
 0&-1  & -1\\
1 &0 & 0 \\
0&0.5&-0.5
\end{array} \right],\quad \varphi(\xi):= \left [ \begin{array}{c}
0\\
0\\
-0.5\xi_2^2
\end{array} \right].
 \end{split}
\end{equation*}
A numerical simulation, see Fig. \eqref{fig.rossler}, confirm that there exists a compact invariant set $\Xi$. In addition, the function $\varphi(\xi)$ is polynomial hence locally lipschitz. \qed
\end{example}

 \begin{figure}[t]
\centering
\subfigure[]{\includegraphics[width=0.48\columnwidth]{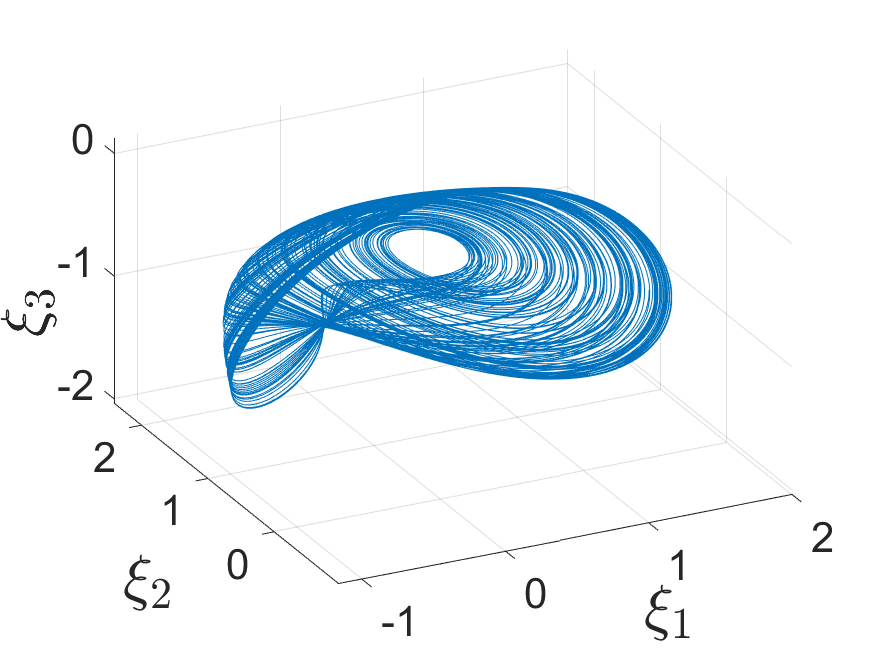}}
\subfigure[]{\includegraphics[width=0.48\columnwidth]{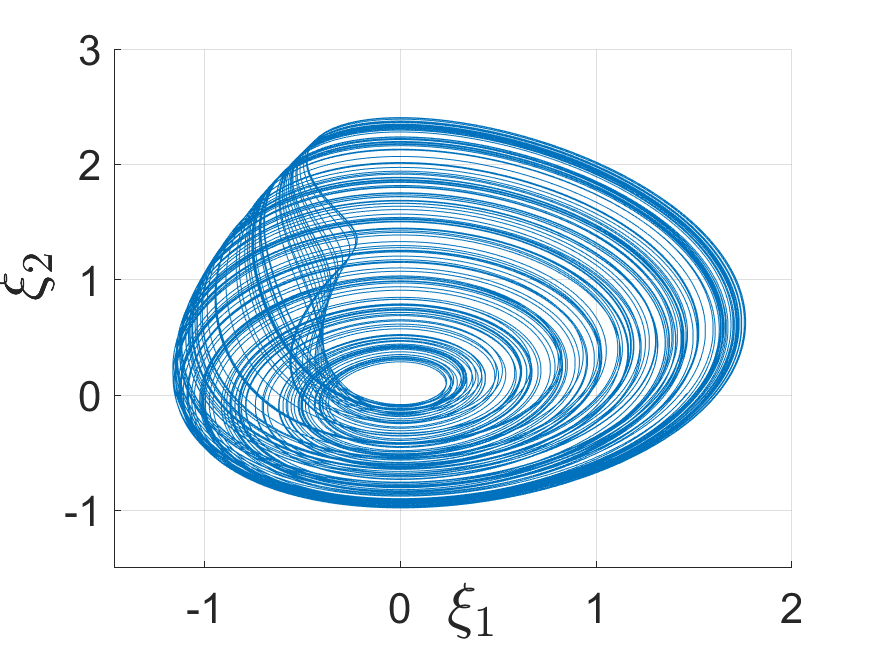}}
 \caption{\textit{R$\ddot{o}$ssler prototype-4 system} from Eq. \eqref{eq.rossler} for $a = b = 0.5$ with initial condition $\xi(0) = [0.1, 0.3,0]^\top $ .}
\label{fig.rossler}
\end{figure}


\subsection{De-masking Approach}

The de-masking is achieved by estimating the extended system $\Sigma_{\bP}$.
For this purpose, an extended observer is designed as
\begin{equation}\label{eq.extended_eta_obsv}
\Sigma_{\bE}:
\; \left\{\begin{split}
 \dot{\hat\bx}&=\bA\hat\bx+\bB u+\bG(\text{sat}_\sigma(\hat\xi))+\bL(\by- \hat{\by})\\
 \hat{\by}& = \bC\hat\bx.
\end{split}
\right.
\end{equation}
where $\sigma_i$ is calculated in \eqref{eq.sigma}.

 Comparing $\Sigma_{\bP}$ and $\Sigma_{\bE}$ and letting $\tilde \bx:= \bx-\hat \bx$ be the estimation error, we have
\begin{equation}\label{eq.extended_system3}
\begin{split}
 \dot{\tilde\bx}&=(\bA-\bL\bC)\tilde\bx+\bG(\xi)-\bG(\text{sat}_\sigma(\hat\xi)).
\end{split}
\end{equation}

The following proposition gives the sufficient condition that $\tilde\bx = 0$ of system \eqref{eq.extended_system3} is exponentially stable.

\begin{proposition}\itshape
      For all $x\in \mathbb{R}^{n_x}$ and $\xi \in \Xi$, suppose that there exist a symmetric positive definite matrix $\bP\in \mathbb{R}^{(n_x+n_\xi)\times (n_x+n_\xi)}$, and  $\bN\in \mathbb{R}^{(n_x+n_\xi)\times n_y}$ satisfying
   \begin{equation}\label{eq.LMI_0}
\begin{split}
\begin{bmatrix}
   \bP\bA+\bA^\top \bP-\bN\bC-\bC^\top \bN^\top+ \begin{bmatrix}
       I&\\
       &0
   \end{bmatrix}        & \bP\\
            \bP& -\ell_\varphi^{-2} I
        \end{bmatrix}<0.
\end{split}
\end{equation}
  Then, the equilibrium point $\tilde\bx=0$ of system \eqref{eq.extended_system3} is exponentially stable.
\end{proposition}
\begin{proof}
Define Lyapunov function $V = \tilde\bx^\top \bP \tilde\bx$.
We have
\begin{equation}\label{eq.d_V}
\begin{split}
 \dot{V}&=\tilde\bx^\top (\bP(\bA-\bL\bC)+(\bA-\bL\bC)^\top \bP)\tilde\bx\\
 &~~~+2\tilde\bx^\top \bP (\bG( \xi)-\bG(\text{sat}_\sigma( \hat\xi))).\\
 \end{split}
\end{equation}

 Since $ \text{sat}_\sigma( \hat\xi),\xi\in \bm\bar\Xi $, according to \eqref{eq.lipschitz_g} and Young’s inequality, we have

\begin{equation*}
\begin{split}
2\tilde\bx^\top \bP(\bG( \xi)-\bG( \text{sat}_\sigma( \hat\xi)))&\leq 2\ell_\varphi  \| \bP \tilde\bx\|\| \tilde \xi \|\\
&\leq \ell_\varphi ^2 \bx^\top \bP\bP\bx + \tilde \xi^\top \tilde \xi.
\end{split}
\end{equation*}

Then \eqref{eq.d_V} becomes
\begin{equation*}
\begin{split}
 \dot{V}&\leq \tilde\bx^\top (\bP(\bA-\bL\bC)+(\bA-\bL\bC)^\top \bP+\ell_\varphi ^2  \bP\bP)\tilde\bx+\tilde \xi^\top \tilde \xi \\
    =    &       \tilde\bx^\top\left[{ \bP(\bA-\bL\bC)+(\bA-\bL\bC)^\top \bP+\ell_\varphi ^2\bP\bP+   \begin{bmatrix}
                I&\\
                &0  \\
            \end{bmatrix} }\right]
                \tilde\bx.\\
\end{split}
\end{equation*}

  Applying the Schur complement to inequality \eqref{eq.LMI_0} with  $Y = \bP\bL $, we obtain
\begin{equation}\label{eq.LMI}
\begin{split}
 { \bP(\bA-\bL\bC)+(\bA-\bL\bC)^\top \bP+\ell_\varphi ^2\bP\bP+   \begin{bmatrix}
                I&\\
                &0  \\
            \end{bmatrix} }  <0,\\
\end{split}
\end{equation}
 implying that there exist some negative definite matrix $Q$ such that
\begin{equation}\label{eq.V_negative}
\begin{split}
\dot V \leq \tilde\bx^\top Q \tilde\bx.
\end{split}
\end{equation}
Recalling that a positive definite quadratic function $\tilde\bx^\top \bP \tilde\bx$ satisfies
\begin{equation}\label{eq.V_posi}
\begin{split}
 \rho_\text{min}(\bP)\tilde\bx^\top \tilde\bx\leq \tilde\bx^\top \bP \tilde\bx\leq \rho_\text{max}(\bP)\tilde\bx^\top \tilde\bx.
\end{split}
\end{equation}

Combining \eqref{eq.V_negative} and \eqref{eq.V_posi}, we can conclude that $\tilde \bx = 0$ of \eqref{eq.extended_system3} is exponentially stable \cite[Theorem 4.10]{khalil2002nonlinear}.

\end{proof}

\subsection{Sufficient Conditions for the Existence of $\bP$ and $\bN$}
In the above subsection, we suppose that there exist $\bP$ and $\bN$. Nevertheless, this may not hold true because $\varphi$ may be too large to ensure \eqref{eq.LMI}.

 In the following, we will give a sufficient condition such that there exist $\bP$ and $\bN$ such that LMI \eqref{eq.LMI_0} holds. Before doing so, a definition and two useful lemmas are given for system
  \begin{equation}\label{eq.system_linear_unobservability}
\begin{split}
  \dot{x}&=A{x}+\phi(x)\\
  y  &= C{x}
\end{split}
 \end{equation}
 with
 \begin{equation}\label{eq.lipschitz_g_2}
\begin{split}
  \|\phi(x_1)-\phi(x_2)\|\leq  \ell_\phi  \|x_1-x_2\|,~~\forall x_1, x_2 \in X.
 \end{split}
\end{equation}

  \begin{definition}\itshape\cite{1998Existence} \textbf{Distance to unobservability:}
  The distance to unobservability of the pair $(A,C)$ for $A\in \mathbb{R}^{n\times n}$ and $C\in \mathbb{R}^{p\times n}$, is defined as the magnitude of the smallest perturbation $(E,F)\in \mathbb{C}^{n\times n}\times \mathbb{C}^{n\times n}$ that makes the pair $(A+E,C+F)$ unobservable. Specifically, this is the quantity
\begin{equation*}
    \delta(A,C) = \inf_{(A+E,C+F) \text{ unobservable}}\|(E,F)\|_2.
\end{equation*}
\end{definition}

 Now, the following lemma provides a method to compute the distance to unobservability.

\begin{lemma}\itshape\cite{1998Existence}
      The distance to unobservability of the pair $(A,C)$ is the minimum with respect to $w$ of the smallest singular value of
    \begin{equation*}
        \begin{bmatrix}
            jwI-A\\
            C
        \end{bmatrix}, w\in R
    \end{equation*}
i.e.,
    \begin{equation*}
        \delta(A,C) = \min_w \varrho_{\text{min}}\begin{bmatrix}
             jwI-A\\
            C
        \end{bmatrix}.
    \end{equation*}
\end{lemma}

In the next Lemma we connect the distance to unobservability with the feasibility of the  LMI \eqref{eq.LMI_0}.

\begin{lemma}\itshape\label{lem.distance}\cite{1998Existence}
    If the distance to unobservability of the pair $(A,C)$ is larger than the
Lipschitz constant $\ell_\phi $ of the nonlinear function   and $A$ is stable, there exists a symmetric
positive definite matrix $P$ and a matrix $L$ such that
\begin{equation*}
\begin{split}
       & P({A}-L{C})+({A}-{L}{C})^\top P+\ell_\phi^2PP+I<0. \\
\end{split}
 \end{equation*}
\end{lemma}

  Finally, the sufficient condition that there exist $\bP$ and $\bN$ for LMI \eqref{eq.LMI_0} is given in the following proposition.
\begin{proposition}\itshape\label{pro.suffi_LMI}
    If the distance to unobservability of the pair $(\bA,\bC)$ is larger than the
 Lipschitz constant $\varphi$, then there exist $\bP$ and $\bN$ such that LMI \eqref{eq.LMI_0} holds.
\end{proposition}
 The proof of this Proposition is done by applying Lemma \ref{lem.distance} and the fact that the existence of $\bP$ and $\bL$ for inequality
\begin{equation*}
     \bP(\bA-\bL\bC)+(\bA-\bL\bC)^\top \bP+\ell_\varphi^2\bP\bP+I<0
\end{equation*}
is sufficient for the existence of $\bP$ and $\bL$ satisfying inequality \eqref{eq.LMI}, which is equivalent to the existence of $\bP$ and $\bN$ of LMI \eqref{eq.LMI_0}.

 \section{Defense of Eavesdropping Attack}
\subsection{Eavesdropping Attack}
 The objective of eavesdropping attack is to access the plant states by eavesdropping the transmitted data. In this section, we assume that an eavesdropper has  access to the communication channel and has full information of the system model, i.e., $\{A,B,C\}$, but no access to the mask \eqref{eq.chaotic_system}.

 It is clear that the attacker can successfully eavesdrop the states of the plant, i.e., $\lim_{t\rightarrow \infty}|{x}-\hat{x}^a|=0$, by the following estimator
  \begin{equation}\label{eq.lin_system_esti_a}
\Sigma_\text{eav}:
\; \left\{\begin{split}
  \dot{\hat{x}}^a&=A{\hat{x}^a}+Bu+\bar L(y-\hat y^a) \\
  \hat y^a  &= C{\hat{x}^a}
\end{split}
\right.
\end{equation}
where $\bar L$ can be any matrix such that $A-\bar L C$ is Hurwitz.

\subsection{Eavesdropping Error Caused by Chaotic Masking}

In this section we aim at showing that in the presence of the proposed chaotic masking protocol, the eavesdroppers estimator $\Sigma_\text{eav}$ would fail to estimate the states of the plant. Since the outputs of a chaotic system cannot be synchronized without the full
knowledge of its specific class and associated internal parameters \cite{2020Resilient}, the chaotic signal added on the sensor measurement is seen as bounded disturbance. Specifically, with chaotic masking, $\Sigma_\text{eav}$ becomes
  \begin{equation}\label{eq.lin_system_esti_a20}
\begin{split}
  \dot{\hat{x}}^a&=A{\hat{x}^a}+Bu+L(y+d-\hat y^a) \\
  \hat y^a  &= C{\hat{x}^a}
\end{split}
\end{equation}
where $\|d\|\leq \varepsilon$ for some $\varepsilon$.

 Then the estimation error dynamic is
 \begin{equation*}
\begin{split}
  \dot{ e }&=(A-LC) e-Ld.
\end{split}
\end{equation*}
Since $d$ is non-vanishing, it is clear that $e (t)$ would not converge to zero.
Specifically, according to  \cite[Corollary 5.2]{khalil2002nonlinear}, the attacker can only guarantee that the estimation error
 \begin{equation*}
\begin{split}
  \|{ e (t)}\|&\leq \frac{2\lambda_\text{max}^2(S)\|L\|}{\lambda_\text{min}(S)}\varepsilon
\end{split}
\end{equation*}
where $S$ is the solution of the Lyapunov equation
$$
S(A-LC)+(A-LC)^{T}S = -I.
$$
Hence  the estimator \eqref{eq.lin_system_esti_a20} would lose its efficacy with a relatively large $\varepsilon$. In other words, we can increase the privacy level of the system by designing chaotic signals such that the amplitude of $d$ is large.

\section{Defense of Replay Attack}

\subsection{Replay Attack}

 The considered replay attack strategy  is given as follows.
 $  \Sigma_\text{replay}: $
\begin{enumerate}
  \item  From time $-\tau$ to time $0$, the attacker records a sequence of sensor measurements $y$.
  \item From time $0$ to time $\tau$, the attacker replays the recorded sensor measurements to tamper the true sensor outputs, i.e.,
\begin{equation}\label{eq.y'}
y^a(t) = y(t-\tau),\quad 0\leq t\leq \tau.
\end{equation}
 Meanwhile, arbitrary harmful attack sequence is injected into the control channel to destroy the plant.
\end{enumerate}

  For simplicity, suppose that $u = -K\hat x$ with $A-BK$ Hurwitz. In addition, here we assume that the control system has run for a long time (already achieved a steady state) before it is attacked, i.e.,  the system and the observer have already reached their steady states before time $-\tau$.

The following Proposition shows the fact that the replay attack is unable to be detected without masking.

\begin{proposition}\itshape\label{eq.pro_replay_without_masking}
 Consider system $\Sigma_p$ and estimator $\Sigma_{e}$, that is, the case without any security mechanism.   Assume that the system and the observer have already reached their steady state at a time $t=-\tau$ for some positive constant $\tau>0$, that is, $x(-\tau)=\hat{x}(-\tau)$. Moreover, assume that $ A-BK-LC $ is Hurwitz. Then, the innovation term $z:= y-\hat{y}$ satisfies $z(t)=0, \forall t\geq0$ under the replay attack $\Sigma_\text{replay}$.
\end{proposition}
\begin{proof}
  The system from time $-\tau$ to zero can be described in the following form
\begin{equation*}
\begin{split}
  \dot{x}^\tau&=A{x}^\tau+Bu^\tau\\
  y^\tau &= C{x}^\tau.
\end{split}
\end{equation*}
The steady-state observer is
\begin{equation}\label{eq.observer_steady}
\begin{split}
\dot{\hat{x}}^\tau &=A\hat{x}^\tau+Bu^\tau+Lz^\tau\\
\end{split}
\end{equation}
with $z^\tau = y^\tau-C\hat x^\tau$ and $u^\tau= -K\hat x^\tau$.

When the observer is a under attack ($0\leq t \leq \tau$),
it becomes
\begin{equation}\label{eq.observer_attaked}
\begin{split}
\dot{\hat{x}}^a &=A\hat{x}^a+Bu^a +Lz^a\\
\end{split}
\end{equation}
with $z^a=y^a-C\hat x^a$ and $u^a = -K\hat x^a$.

 Subtracting \eqref{eq.observer_steady} with \eqref{eq.observer_attaked}, there is
\begin{equation*}
\begin{split}
\dot{\hat{x}}^\tau-\dot{\hat{x}}^a &=(A-BK-LC)(\hat{x}^\tau-\hat{x}^a) +L (y^\tau-y^a).
\end{split}
\end{equation*}

 Since $y_k^\tau=y_k^a$ (by $\eqref{eq.y'}$) and $\Lambda:= A-BL-LC $ is Hurwitz, we have
\begin{equation*}
\begin{split}
\hat{x}^\tau-\hat{x}^a = e^{\Lambda t}(\hat{x}^\tau(0)-\hat{x}^a(0)).
\end{split}
\end{equation*}

Consequently, the innovation of the attacked system is
\begin{equation*}
\begin{split}
z^a&=y^a-C\hat x^a=y^\tau-C \hat x^\tau+Ce^{\Lambda t}(\hat{x}^\tau(0)-\hat{x}^a(0)).
\end{split}
\end{equation*}

  Since the observer has already reached its steady state before $-\tau$, we have $y^\tau-C \hat x^\tau = 0$. In addition, since the system $\Sigma_p$ is linear, it stays at an equilibrium point in steady state. Consequently, we have
 $\hat{x}^\tau(0)=\hat{x}^a(0)$.
Further considering the fact that $ \Lambda$ is Hurwitz, $e^{\Lambda t}$ is bounded, and hence $e^{\Lambda t}(\hat{x}^\tau(0)-\hat{x}^a(0))=0$.
As a result, the innovation under replay attack equals to $0$ and the attack cannot be detected.

\end{proof}

\subsection{Invalidating the Replay Attack Strategy}
 In the above subsection, the replay attack strategy relies on the fact that the  physical plant stays at an equilibrium point in steady state. With chaotic masking, the extended system $\Sigma_{\bP}$ contains chaotic dynamics and has a more complex steady-state behavior, which is no longer an equilibrium, periodic oscillation, or almost-periodic oscillation. Consequently, any replay attack strategy will increase, at least during a transient time, the innovation term $\bz:= \by-\hat{\by}$. This fact is formalized in the next proposition.

  \begin{proposition}\label{eq.pro_replay_with_masking}
 {\itshape Consider the extended system $\Sigma_{\bP}$ and the extended estimator $\Sigma_{\bE}$, that is, the case with chaotic masking. Assume that the estimation error of the extended system and the extended observer has reached its steady state before time $-\tau$, that is,} $ \bx(-\tau)=\hat{\bx}(-\tau)$. {\itshape Consider the replay attack $\Sigma_\text{replay}$, then, the innovation} $\bz:= \by-\hat{\by}$ {\itshape cannot be zero for all $t$ such that $0\leq t \leq \tau$.}
\end{proposition}
\begin{proof}
 The extended system from $-\tau$ to zero can be described in the following form
\begin{equation*}
\begin{split}
 \dot{\bx}^{\tau}&=\bA\bx^{\tau}+\bB(-K x^{\tau})+\bG(\xi^{\tau})\\
  \by^{\tau}  &= \bC\bx^{\tau}
  \end{split}
\end{equation*}
The steady-state observer is
\begin{equation*}
\begin{split}
 \dot{\hat \bx}^\tau&=\bA\hat\bx^\tau+\bB(-K x^{\tau})+\bG(\text{sat}_\sigma(\hat\xi^{\tau}))+\bL \bz^\tau\\
  \hat{ \by}^\tau  &= \bC\hat\bx^\tau\\
   \end{split}
\end{equation*}
with $\bz^\tau =\by^\tau- \hat{ \by}^\tau$.

   When the observer is under replay attack ($0\leq t \leq \tau$), it becomes
\begin{equation}\label{eq.extend_observer_attaked}
\begin{split}
 \dot{\hat \bx}^a&=\bA\hat\bx^a+\bB(-Kx^{a})+\bG(\text{sat}_\sigma(\hat\xi^{a}))+\bL \bz^a\\
  \hat{ \by}^a  &= \bC\hat\bx^a\\
  \end{split}
 \end{equation}
with $\bz^a =\by^a- \hat{ \by}^a=\by^\tau- \hat{ \by}^a$.

The following proof is by contradiction.
 If $\bz^a\equiv 0$ for all $t$ such that $0\leq t \leq \tau$, the replay attack does not affect \eqref{eq.extend_observer_attaked} and we have $\hat\bx^a(t) =\hat\bx(t)$ for $0\leq t \leq \tau$.
Consequently,
\begin{equation*}
\begin{split}
\bz^a \equiv 0\Leftrightarrow  \by^\tau- \hat{ \by}^a\equiv 0 \Leftrightarrow \bC\bx^\tau \equiv \bC\hat\bx^a\Leftrightarrow \bC\bx^\tau \equiv \bC\hat\bx.
 \end{split}
 \end{equation*}

   Since the estimator has already reached steady state, we have $\hat\bx = \bx$, and hence we can deduce $\bC\bx^\tau(t) \equiv \bC\bx(t), 0\leq t \leq \tau$, implying (since $\bx$ is instantaneously observable)
\begin{equation*}
\bx^\tau(t) \equiv \bx(t),~~ 0\leq t \leq \tau,
 \end{equation*}
which contradicts with the fact that $\bx$ is chaotic, that is, non-periodic.

\end{proof}

\section{Defense of Stealthy Data Injection Attack}
 \subsection{Stealthy Data Injection Attack}
A stealthy data injection attacker injects attack sequence $\{a\}$ into the sensor-estimator channel to ruin the estimator $\Sigma_{e}$ while keeping stealthy. The adjective ``stealthy'' refers to the fact that the sequence $\{a\}$ can be designed to be undetected by common innovation-based detectors such as \eqref{eq.detector}.

   For simplicity, we suppose  $u = 0$.  To launch stealthy data injection attack, we suppose that the attack can get the system information $\{A,B,C,L\}$ and can inject arbitrary $\{a\}$ into the sensor measurement.
With stealthy data injection attack, the unprotected estimator $\Sigma_{e}$ becomes
\begin{equation}\label{eq.lin_system_esti_a2}
\begin{split}
  \dot{\hat{x}}^a&=A{\hat{x}}^a+L(y^a-\hat y^a) \\
  \hat y ^a &= C{\hat{x}}^a
\end{split}
\end{equation}
where $y^a = y+a$ with $a$ is the attack injected by adversaries. Define estimation difference $\Delta \hat{x} :={\hat{x}}^a-{\hat{x}} $ and innovation difference $\Delta z := z^a-z$. Comparing $\Sigma_{e}$ with \eqref{eq.lin_system_esti_a2}, there is
\begin{equation}\label{eq.Delta_lin_system_esti}
\begin{split}
  \Delta\dot{\hat{x}}&=(A-LC)\Delta{\hat{x}}+La  \\
  \Delta z &=-C\Delta{\hat{x}}+a
\end{split}
\end{equation}
 According to the triangular inequality $\| z^a\|\leq \|z\|+ \|\Delta z\|$, if $\|\Delta z\|$ is small, the detector cannot distinguish between $z^a$ and $z$ with high probability. Consequently, the detection rate of detector \eqref{eq.detector} is close to its false alarm rate. Hence the attack is assumed to be stealthy if
\begin{equation}\label{eq.stealthy}
\begin{split}
  \|\Delta z\|\leq M.
\end{split}
\end{equation}
Based on \eqref{eq.stealthy}, stealthy attack against remote estimator is designed as
\begin{equation}\label{eq.stealthy_attack}
\begin{split}
\Sigma_\text{FDI}: \;  a = C\Delta{\hat{x}}+ \varphi.
 \end{split}
 \end{equation}
It is noted that, without chaotic masking, $\Delta{\hat{x}}$ can be obtained by  \eqref{eq.Delta_lin_system_esti} with $\Delta{\hat{x}}(0)=0$. By choosing $\varphi$ as any time-varying functions satisfying $ \|\varphi \|\leq M$, $\Sigma_\text{FDI}$ is stealthy in the sense of \eqref{eq.stealthy}. In addition, by specific design of $\varphi$, one can regulate the estimation difference to some prescribed values \cite{chen2023optimal} or maximize it \cite{ning2023design}.

 \subsection{Destroying the Stealthiness of the Data Injection Attack}
In this subsection, we will show that, with chaotic masking, $\Sigma_\text{FDI}$ is no more stealthy.

 With chaotic masking, the extended estimator under stealthy attack becomes
  \begin{equation}\label{eq.extended_system_observer_a}
\begin{split}
 \dot{\hat \bx}^a&=\bA\hat\bx^a+\bG(\text{sat}_\sigma(\hat\xi^a))+\bL(\by^a- \hat{ \by}^a)\\
  \hat{ \by}^a  &= \bC\hat\bx^a\\
\end{split}
\end{equation}
Let $\bz:=\by- \hat{ \by}$, $\Delta \hat \bx:=\hat \bx^a-\hat \bx$ and $\Delta \bz := \bz^a-\bz $.
 According to $\Sigma_{\bE}$ and \eqref{eq.extended_system_observer_a}, we have
 \begin{equation}\label{eq.extended_system_Delta_Z}
\begin{split}
\Delta \dot {\hat \bx}&=(\bA-\bL \bC)\Delta\hat \bx+\bG(  \text{sat}_\sigma(\hat\xi^a))-\bG(\text{sat}_\sigma(\hat\xi))+\bL a\\
 \Delta \bz &= -\bC \Delta \hat \bx +a
 \end{split}
\end{equation}
Since the term $\bC \Delta \hat \bx$ contains both $\Delta {\hat{x}}$ and $\Delta {\hat\xi}$, it is unlikely to be offset by $\Delta {\hat{x}}$ in $\Sigma_\text{FDI}$. As a result, attack $\Sigma_\text{FDI}$ cannot guarantee that $\|\Delta \bz\| = \|-\bC \Delta \hat \bx+\Delta {\hat{x}}+\varphi\|\leq M$. In other words, $\Sigma_\text{FDI}$ is no more sealthy.

 In addition, since the system \eqref{eq.extended_system_Delta_Z} is input-to-output stable from $a$ to $\Delta \bz$ (this fact can be verified by  \cite[Theorem 5.1]{khalil2002nonlinear} by taking $V =  \Delta \hat \bx^\top \bP  \Delta \hat \bx $), there exist some $\kappa$ such that if $\|a\|\leq \kappa$, we have $\|\Delta \bz\| \leq M$. In other words,  $a$ may still be designed small enough to make sure the attack is stealthy. However, the supremum of $\kappa$ may be too small to launch any significant attacks, and,  it is hard for a potential attacker to estimate the supremum of $\kappa$ without any information of $\Phi$.

\section{Simulation Verification}
 The system used in the simulation is borrowed from \cite{franze2021resilient}, which is a model of an longitudinal axis motion of the $B747-100/200$ aircraft, whose states are the angle of altitude ($x(1)$), the pitch angle  ($x(2)$), the pitch rate ($x(3)$), and the true airspeed ($x(4)$), respectively.  The system matrix, input matrix, and output matrix are
\begin{equation*}
\begin{split}
&A=\begin{bmatrix}
             -0.5717& 0 &1.005 &-0.0006\\
        0& 0& 1 &0\\
 -0.1049 & 0 &-0.6803&  0.0002\\
-4.6726 &-9.7942 &-0.1463& -0.0062\\
        \end{bmatrix}.
\end{split}
\end{equation*}
\begin{equation*}
\begin{split}
&B=\begin{bmatrix}
            0& 0 &0\\
       0  &0 &0\\
-1.5539 &0.0154& -0.1556\\
0 &1.3287 &0.2000\\
        \end{bmatrix}, C=\begin{bmatrix}
        1 &0 &0& 1\\
      0 &0& 1& 1
        \end{bmatrix}.
\end{split}
\end{equation*}

 To mask the sensor measurements, we use the R$\ddot{o}$ssler prototype-4 system with
\begin{equation*}
\begin{split}
  \Lambda = \begin{bmatrix}
        1 &0 &0\\
      0 &2& 1
        \end{bmatrix} .
 \end{split}
\end{equation*}
For solving LMI \eqref{eq.LMI_0}, we test the sufficient condition in proposition \ref{pro.suffi_LMI}, firstly.
The constant $\ell_\varphi $ in \eqref{eq.lipschitz_g} is estimated as $\ell_\varphi  =2.5$ since
 \begin{equation}\label{eq.estimate_phi}
\begin{split}
&\left\|\frac{\partial \varphi(\xi)}{\xi}\right\| =\left\|\begin{bmatrix}
         0 &0&0\\
      0&0& 0\\
      0&y&0
        \end{bmatrix}\right\| \leq 2.5,
\end{split}
\end{equation}
according to Fig. \ref{fig.rossler} (b).
 \begin{figure}[h]
\centering
\subfigure[]{\includegraphics[width=0.4\columnwidth]{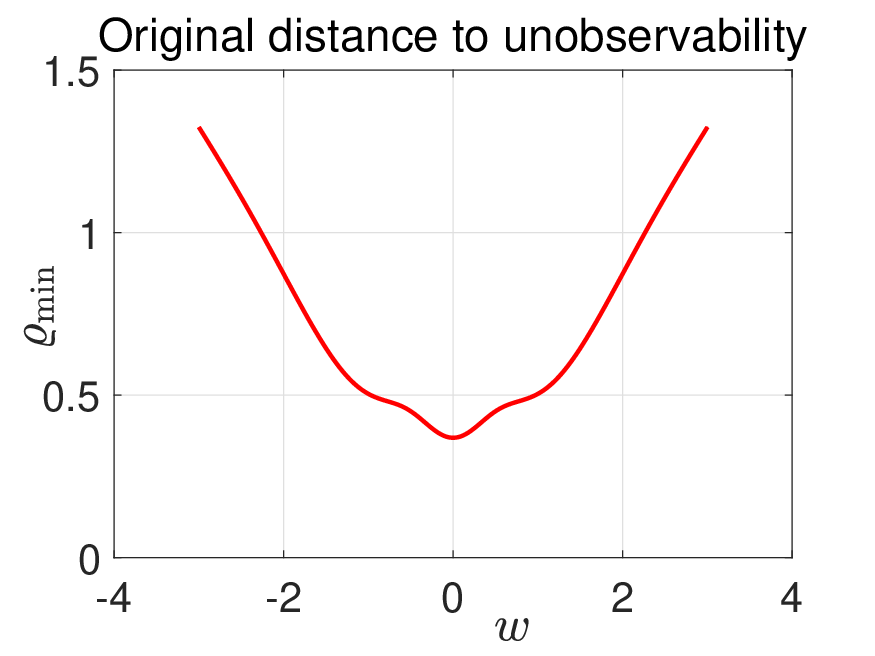}}
\subfigure[]{\includegraphics[width=0.4\columnwidth]{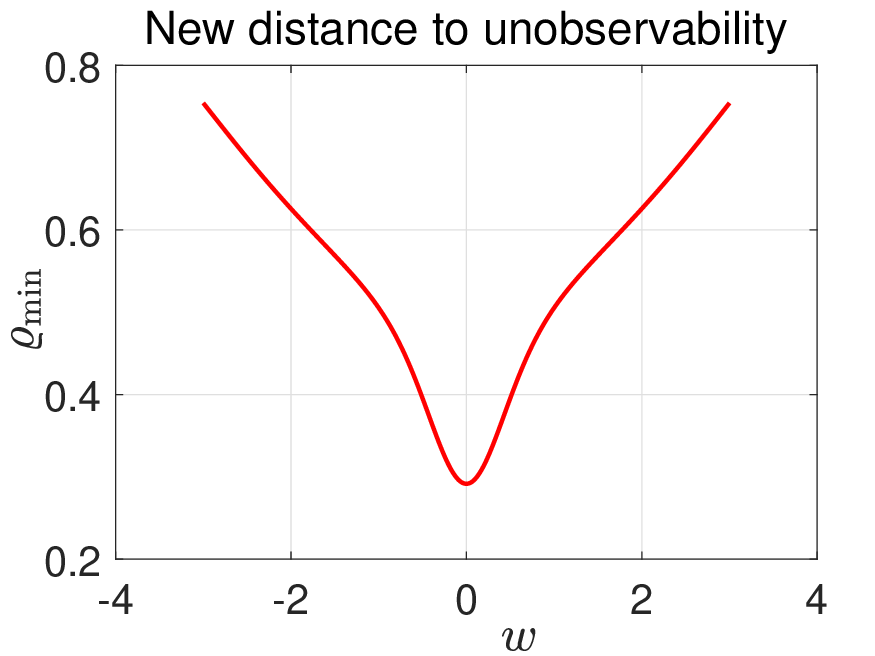}}
\caption{Distance to unobservability}
\label{fig.distance}
\end{figure}

The distance to unobservability of $(\bA,\bC)$ is about $0.4$ (see Fig. \ref{fig.distance}.(a)), which is smaller than $\varphi$, not satisfying the sufficient condition in Proposition \ref{pro.suffi_LMI}. To deal with this problem, we apply the following coordinate transformation to the chaotic system, $\xi' = \text {diag} (1,1,1/\beta)\xi$, leading to the following chaotic system
 \begin{equation*}
\begin{split}
 \dot\xi' = \left [ \begin{array}{ccc}
 0&-1  & -\beta \\
1 &0 & 0 \\
0&0.5\beta&-0.5
\end{array} \right]\xi'+ \left [ \begin{array}{c}
0\\
0\\
-\frac{0.5}{\beta}\bar y^2
\end{array} \right].
 \end{split}
\end{equation*}

  Taking $\beta = 100$, according to \eqref{eq.estimate_phi}, $\varphi$ can be estimated as $\varphi = 0.025$, and, in the new coordinate, the distance to unobservability of $(\bA,\bC)$ is about $0.3$ (see Fig. \ref{fig.distance} (b)), which satisfies the sufficient condition in proposition \ref{pro.suffi_LMI}.

 Then, by solving \eqref{eq.LMI_0}, $\bP$ and $\bN$ are obtained and $\bL$ is calculated as
 \begin{equation*}
\bL=\bP^{-1}\bN=\begin{bmatrix}
 -147.5759  &192.5241\\
  -83.7116  &107.5186\\
  -31.3697  &  4.5360\\
  -10.1254  &  6.4027\\
  -25.5599  & 17.7149\\
   -8.3158 &   2.4874\\
  175.7012 &-197.6720
  \end{bmatrix}.
   \end{equation*}

  \begin{figure}[ht]
\centering
\subfigure[]{\includegraphics[width=0.48\columnwidth]{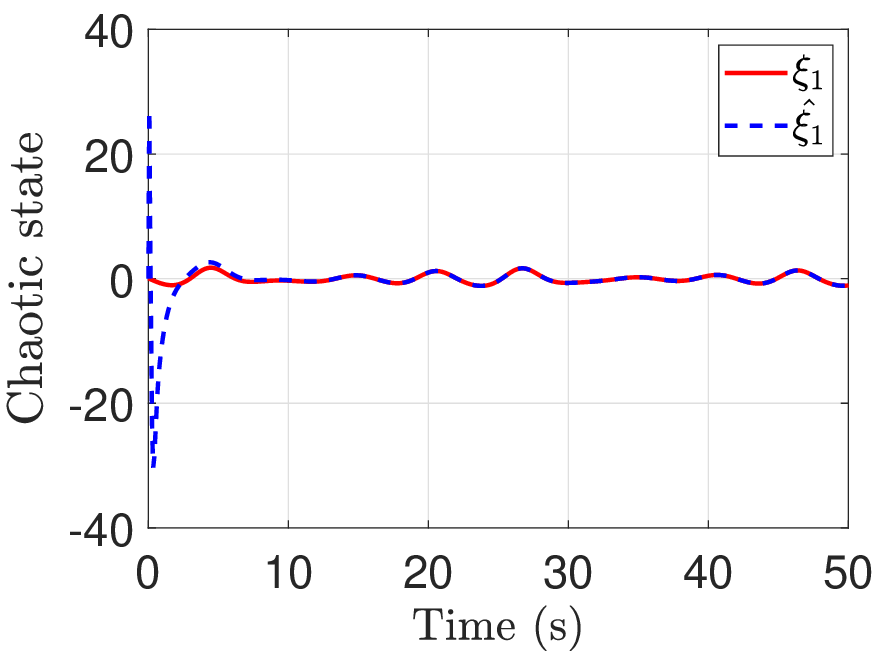}}
\subfigure[]{\includegraphics[width=0.48\columnwidth]{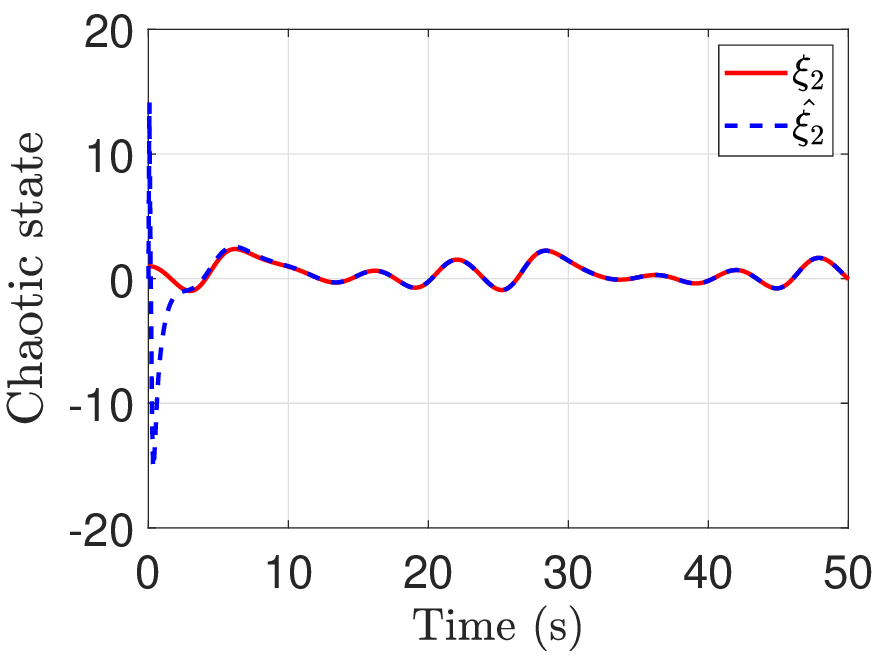}}
\subfigure[]{\includegraphics[width=0.48\columnwidth]{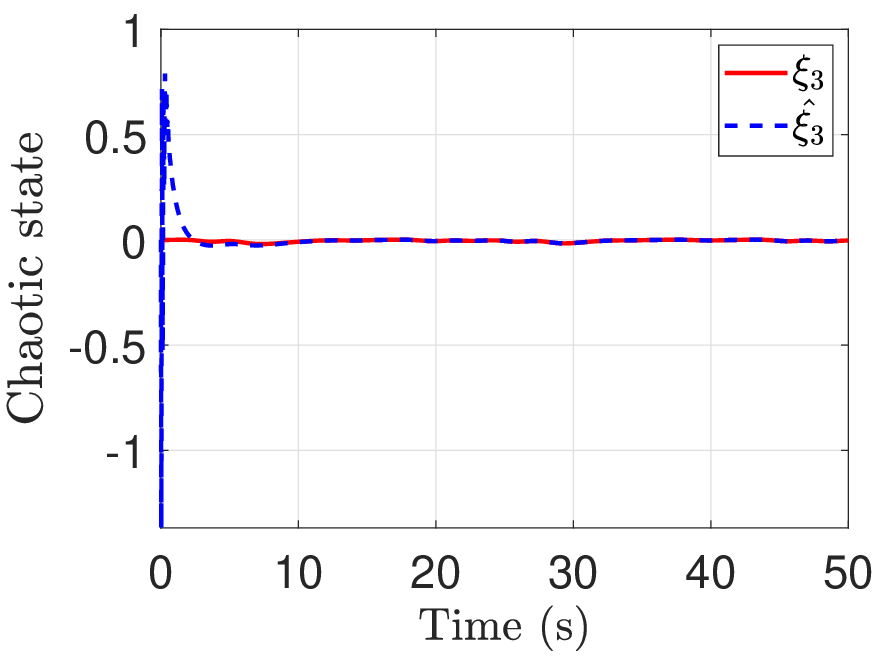}}
\subfigure[]{\includegraphics[width=0.48\columnwidth]{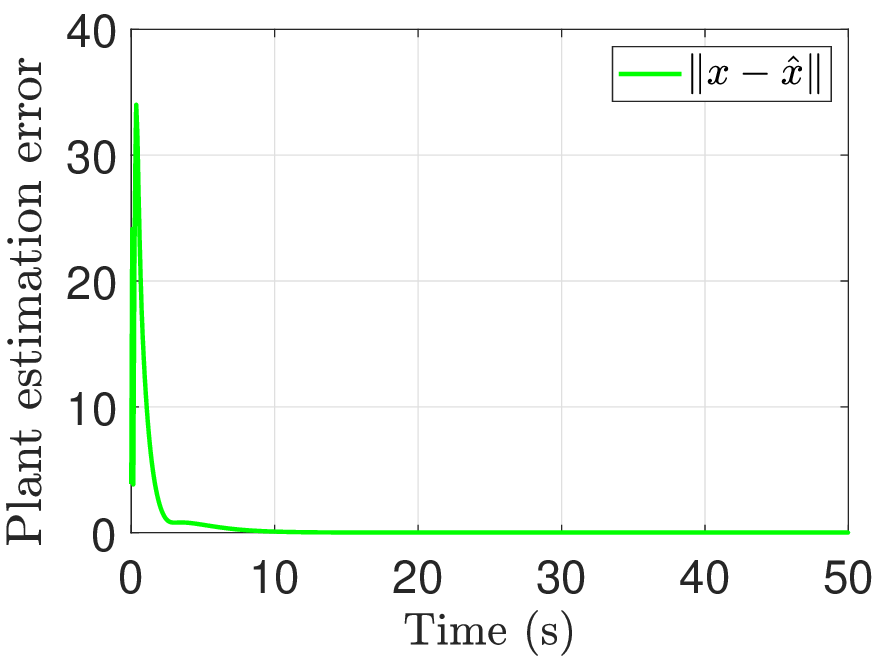}}
\caption{Simulation results without attack.  (a)-(c) The chaotic signal and their corresponding estimations. (d) The norm of estimation error of the physical plant.}
\label{fig.no_attack}
\end{figure}

\begin{figure}[ht]
\centering
\includegraphics[width=0.75\columnwidth]{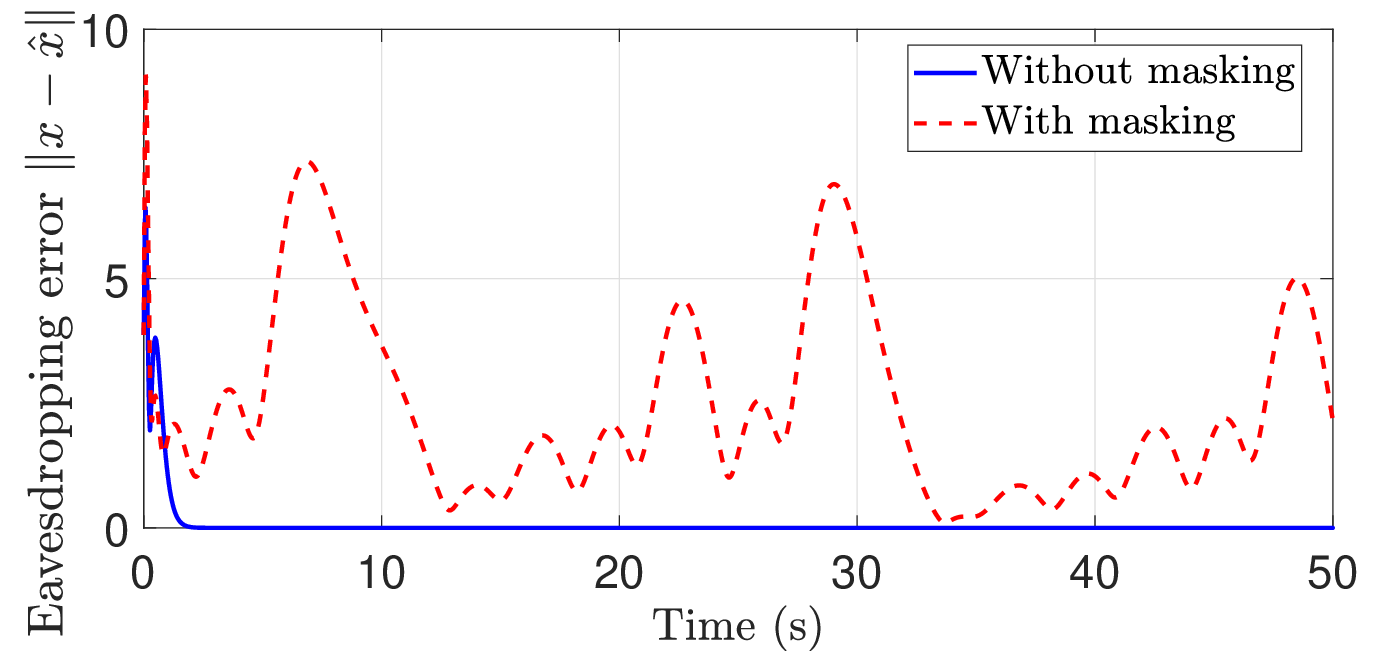}
 \caption{The eavesdropping error with and without masking protocol.}
\label{fig.eva}
\end{figure}

\begin{figure}[ht]
\centering
\includegraphics[width=0.75\columnwidth]{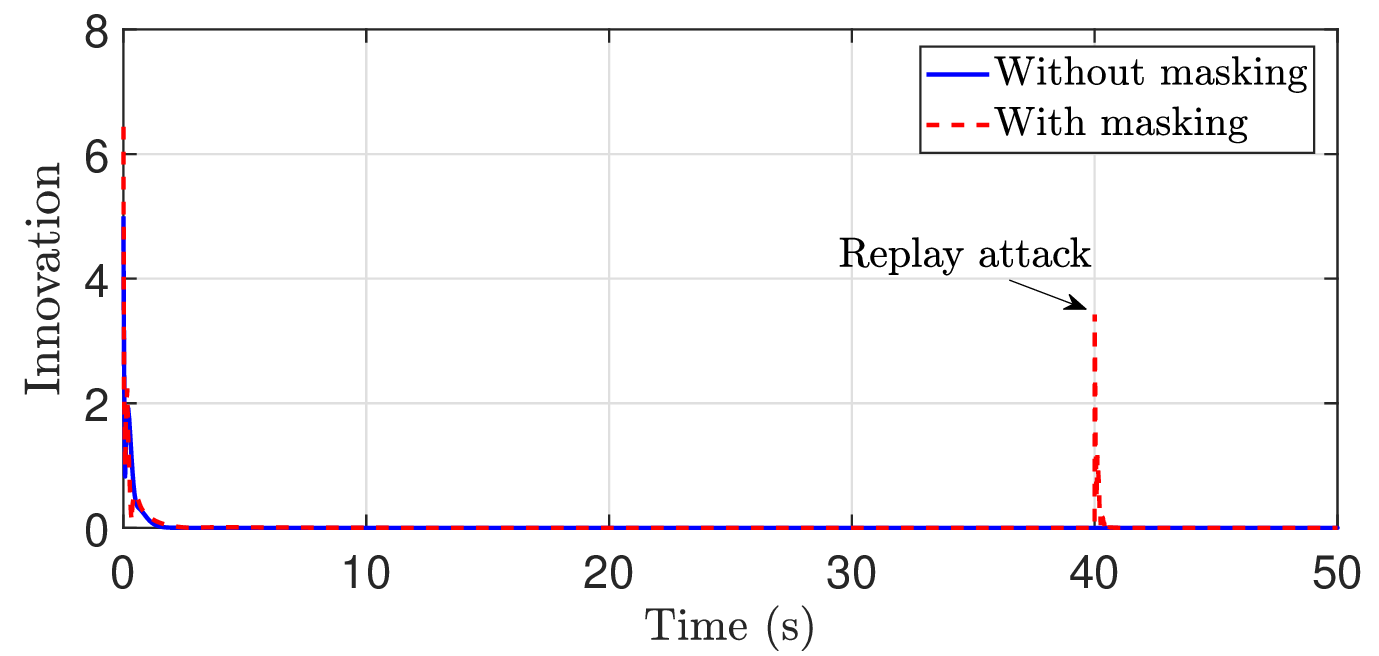}
\caption{The innovation under replay attack.}
\label{fig.replay}
\end{figure}

\begin{figure}[ht]
\centering
\includegraphics[width=0.75\columnwidth]{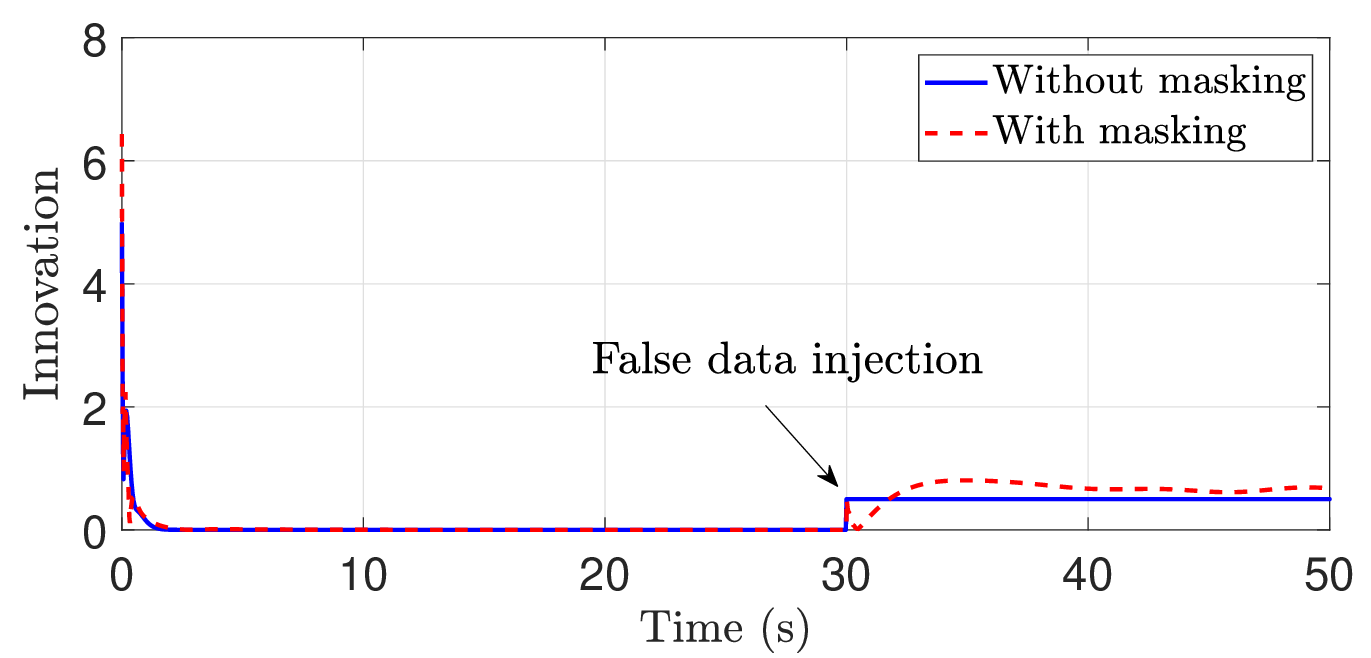}
 \caption{The innovation under FDI attack.}
\label{fig.fdi}
\end{figure}

 In addition, to stabilize the system, the control input is designed as $u = -K\hat x$ with
 \begin{equation*}
K= \begin{bmatrix}
    0.0703  & -0.2808 &  -0.3223 &  -0.0021\\
   -3.4397  & -7.2033 &  -0.1029  &  0.3634\\
   -0.5106  & -1.1128  & -0.0483  &  0.0545\\
  \end{bmatrix},
  \end{equation*}
and, for comparison, $L$ of the original estimator $\Sigma_{e}$ is chosen as
\begin{equation*}
L=\begin{bmatrix}
 -0.1860  & -2.2246 &  -3.6838 &   6.3543\\
  -17.6414  &  1.4233 & -10.1427 &  25.9750\\
  \end{bmatrix}.
   \end{equation*}

The simulation result when there is no attack is given in Fig. \ref{fig.no_attack}, where both chaotic dynamic and physical plant are well estimated,  verifying the effectiveness of the proposed method.

 The simulation result under eavesdropping attack is given in Fig. \ref{fig.eva}, where $\Sigma_\text{eav}$ is used for estimating protected $\Sigma_{\bP}$ and unprotected system $\Sigma_p$. From Fig. \ref{fig.eva} we can see that if the sensor measurements are not protected, the eavesdropping error converges to zero, where as the eavesdropping error is bounded but does not converge if the sensor measurements are masked.

The simulation result under replay attack is given in Fig. \ref{fig.replay}, where the attack is launched at time point $40s$. The replay attack $\Sigma_\text{replay}$ is launched for the protected system $\Sigma_{\bP}$ and unprotected system $\Sigma_p$. We can find that, if without masking, the innovation maintains $0$ when replay attack launches, where as the innovation jumps with masking, implying that the replay can be detected with the masking protocol and the replay attack strategy is no longer valid.

 The simulation result under FDI attack is given in Fig. \ref{fig.fdi}, where $M$ is supposed as $0.5$ and the attack is launched at time point $30s$. The FDI attack $\Sigma_\text{FDI}$ is launched for both protected $\Sigma_{\bP}$ and unprotected system $\Sigma_p$. If the measurements are not protected, the innovation difference always less than $0.5$, which satisfies the stealthy condition \eqref{eq.stealthy}. However, if the  sensor measurements are masked, the innovation difference beyond $0.5$, and thus FDI attack $\Sigma_\text{FDI}$ is no longer stealthy.

\section{Conclusion}
 In this paper, a chaotic masking protocol has been proposed to secure the remote estimator. Since an extended estimator was designed to estimate both chaotic states and physical system states, the masking protocol does not require extra communication to synchronize the masking-damasking pair and its effect can be perfectly removed in steady state. To verify the effectiveness of the proposed chaotic masking protocol, three types of attacks were considered. Future works will focus on extending the
proposed framework to mask non-linear systems and to protect whole control system loops.

\bibliography{ifacconf}

\end{document}